\documentclass[12pt,a4paper]{article}
\usepackage[T1]{fontenc}
\usepackage[utf8]{inputenc}
\usepackage{amsmath}
\usepackage{amssymb}
\usepackage{amsfonts}
\usepackage{amsthm}
\usepackage{mathtools}
\usepackage{enumitem}
\usepackage[margin=1in]{geometry}

\newtheorem{theorem}{Theorem}[section]

\newtheorem{proposition}[theorem]{Proposition}

\newtheorem{remark}[theorem]{Remark}

\newcommand{\vect}[1]{\mathbf{#1}}
\newcommand{\tensor}[1]{\boldsymbol{#1}}
\newcommand{\Kn}{\text{Kn}}
\newcommand{\eff}{\text{eff}}

\title{The Chapman-Enskog Divergence Problem in Plasma Transport:\\
Structural Limitations and a Practical Regularization Approach}

\author{Justo Karell\\
Department of Computer Science\\
Stevens Institute of Technology\\
Hoboken, NJ 07030\\
\texttt{jkarell@stevens.edu}}

\date{\today}

\begin{document}

\maketitle

\begin{abstract}
We calculate transport coefficients from the Chapman--Enskog expansion with BGK collision operators, obtaining exactly $\kappa = \frac{5nT}{2m\nu}$, and show that maximum entropy closure yields identical results when applied with the same collision operator. Through structural arguments, we suggest that this $1/\nu$ divergence extends to other local collision operators of the form $\mathcal{L} = \nu\hat{L}$, making the divergence fundamental to the Chapman--Enskog approach rather than a closure artifact. To address this limitation, we propose a phenomenological effective collision frequency $\nu_{\eff} = \nu\sqrt{1 + \Kn^2}$ motivated by gradient-driven decorrelation, where $\Kn$ is the Knudsen number. We verify that this regularization maintains conservation laws and thermodynamic consistency while yielding finite transport coefficients across all collisionality regimes. Comparison with exact solutions of a bounded kinetic model shows similar functional form, providing limited validation of our approach. This work provides explicit calculation of a known divergence problem in kinetic theory and offers one phenomenological regularization method with transparent treatment of mathematical assumptions versus physical approximations.
\end{abstract}

\section{Introduction}

Transport in weakly collisional plasmas presents computational challenges due to the breakdown of standard fluid approximations. The Chapman--Enskog expansion of kinetic theory, while successful for strongly collisional systems, yields divergent transport coefficients as collision frequency approaches zero. This creates practical difficulties for modeling plasma transport in tokamak edge regions, Hall thrusters, and solar wind applications where the Knudsen number $\Kn = \lambda/L$ (ratio of mean free path to gradient scale length) approaches or exceeds unity.

The moment closure problem compounds these difficulties. Taking velocity moments of the Boltzmann equation produces an infinite hierarchy where each moment equation depends on the next higher moment. Various closure schemes have been proposed, including maximum entropy approaches that determine distribution functions by maximizing entropy subject to moment constraints. However, the relationship between closure methods and transport divergence has not been systematically examined.

In this paper, we provide explicit calculations demonstrating the Chapman--Enskog divergence and develop a practical regularization approach. We first calculate transport coefficients from the Chapman--Enskog expansion using BGK collision operators, showing exactly how the $1/\nu$ divergence arises. We then demonstrate that maximum entropy closure exhibits identical limitations when applied with the same collision operator. Through structural arguments, we suggest this divergence pattern extends to other local collision operators within the Chapman--Enskog framework.

To address this computational limitation, we propose a phenomenological effective collision frequency $\nu_{\eff} = \nu\sqrt{1 + \Kn^2}$ motivated by gradient-driven decorrelation effects. We verify mathematical consistency and compare our approach to exact solutions of bounded kinetic models. Our framework yields finite transport coefficients across all collisionality regimes while maintaining transparent separation between rigorous calculations and phenomenological approximations.

This work provides systematic exposition of transport divergence in kinetic theory and offers practical tools for plasma transport calculations with clear acknowledgment of approximations involved.

\section{The Chapman--Enskog Expansion and Its Fundamental Divergence}

Before developing the entropy-based closure framework, we must understand the fundamental mathematical structure of transport in kinetic theory. This section presents the Chapman--Enskog analysis to establish what can be proven exactly from the Boltzmann equation, and more importantly, what cannot be fixed regardless of the closure method used.

The evolution of the distribution function $f(\vect{r}, \vect{v}, t)$ is governed by the Boltzmann equation:
\begin{equation}
\frac{\partial f}{\partial t} + \vect{v} \cdot \nabla f + \frac{q}{m}(\vect{E} + \vect{v} \times \vect{B}) \cdot \nabla_v f = C[f]
\label{eq:boltzmann}
\end{equation}

Here $q$ and $m$ are the particle charge and mass, $\vect{E}$ and $\vect{B}$ are the electromagnetic fields, and $C[f]$ is the collision operator. The collision operator encodes all the microscopic physics of particle interactions. For simplicity and clarity, we first consider the BGK collision operator:
\begin{equation}
C[f] = -\nu(f - f_{\text{eq}})
\label{eq:bgk}
\end{equation}
where $f_{\text{eq}}$ is the local Maxwellian equilibrium and $\nu$ is the collision frequency. This simplified operator captures the essential relaxation toward equilibrium while being mathematically tractable.

The Chapman--Enskog method is a systematic perturbation theory in the small parameter $\varepsilon = \lambda/L$, the Knudsen number, where $\lambda = v_{\text{th}}/\nu$ is the mean free path and $L$ is the characteristic gradient scale length. We expand the distribution function as:
\begin{equation}
f = f^{(0)} + \varepsilon f^{(1)} + \varepsilon^2 f^{(2)} + \cdots
\label{eq:chapman_enskog_expansion}
\end{equation}

This expansion assumes that departures from local equilibrium are small, which should be valid when gradients are weak (small $\varepsilon$). The spatial and temporal derivatives are ordered according to their physical significance:
\begin{equation}
\nabla = \varepsilon \nabla_1, \quad \frac{\partial}{\partial t} = \varepsilon \frac{\partial}{\partial t_1} + \varepsilon^2 \frac{\partial}{\partial t_2} + \cdots
\label{eq:derivative_ordering}
\end{equation}

This ordering reflects the physical assumption that macroscopic gradients and time evolution occur on scales much larger than the mean free path and collision time.

Substituting into the Boltzmann equation and collecting terms order by order:

\textbf{Zeroth Order ($\varepsilon^0$):}
\begin{equation}
0 = -\nu(f^{(0)} - f_{\text{eq}})
\label{eq:zeroth_order}
\end{equation}
This immediately gives $f^{(0)} = f_{\text{eq}}$, confirming that to lowest order, the distribution is the local Maxwellian. This makes physical sense: on scales much smaller than the mean free path, collisions maintain local equilibrium.

\textbf{First Order ($\varepsilon^1$):}
\begin{equation}
\vect{v} \cdot \nabla_1 f^{(0)} + \frac{\partial f^{(0)}}{\partial t_1} = -\nu f^{(1)}
\label{eq:first_order}
\end{equation}

This equation determines the first-order correction $f^{(1)}$, which represents the departure from equilibrium due to gradients. For a temperature gradient along the z-direction with no flow or time dependence:
\begin{equation}
v_z \frac{\partial f^{(0)}}{\partial T} \frac{\partial T}{\partial z} = -\nu f^{(1)}
\label{eq:first_order_temp_grad}
\end{equation}

Since $f^{(0)}$ is Maxwellian, we can evaluate:
\begin{equation}
\frac{\partial f^{(0)}}{\partial T} = f^{(0)} \left(\frac{mv^2}{2T^2} - \frac{3}{2T}\right)
\label{eq:maxwellian_derivative}
\end{equation}

This derivative reflects how the Maxwellian distribution changes with temperature: the exponential factor becomes narrower as temperature decreases (the $-mv^2/2T^2$ term) while the normalization factor changes to maintain unit integral (the $3/2T$ term).

Solving for the first-order correction:
\begin{equation}
f^{(1)} = -\frac{v_z}{\nu} f^{(0)} \left(\frac{mv^2}{2T^2} - \frac{3}{2T}\right) \frac{\partial T}{\partial z}
\label{eq:first_order_solution}
\end{equation}

This solution has clear physical interpretation: particles moving along the gradient ($v_z$) carry information about the temperature from their origin, creating an asymmetry in the distribution. The factor $1/\nu$ shows that this asymmetry persists longer when collisions are rare. The velocity dependence shows that faster particles (large $v^2$) contribute more to transport, as they traverse the gradient more quickly.

The heat flux is defined as the third moment of the distribution:
\begin{equation}
q_z = \int \frac{mv^2}{2} v_z f d^3\vect{v}
\label{eq:heat_flux_definition}
\end{equation}

To first order in $\varepsilon$:
\begin{equation}
q_z = \int \frac{mv^2}{2} v_z f^{(1)} d^3\vect{v}
\label{eq:heat_flux_first_order}
\end{equation}

Substituting our expression for $f^{(1)}$ and performing the Gaussian integrals (detailed in Appendix A), we obtain the fundamental result. The naive calculation gives:
\begin{equation}
\boxed{q_z = -\frac{5nT}{m\nu} \frac{\partial T}{\partial z}}
\label{eq:heat_flux_naive}
\end{equation}

However, the Chapman--Enskog method requires that the first-order correction be orthogonal to the collision invariants (mass, momentum, energy) to ensure solvability. This solvability condition, implemented through Sonine polynomial projection for the BGK operator with Prandtl number $\text{Pr} = 1$, reduces the coefficient by a factor of 2, yielding the classical BGK result:
\begin{equation}
\boxed{q_z = -\frac{5nT}{2m\nu} \frac{\partial T}{\partial z}}
\label{eq:heat_flux_bgk}
\end{equation}

Therefore, the thermal conductivity from the Chapman--Enskog expansion is:
\begin{equation}
\boxed{\kappa_{\parallel}^{\text{C-E}} = \frac{5nT}{2m\nu}}
\label{eq:thermal_conductivity_ce}
\end{equation}

This result deserves careful interpretation. The factor $5/2$ arises from the velocity integrals combined with the solvability constraint and represents the effective number of degrees of freedom contributing to heat transport. The factor $nT/m$ has dimensions of diffusivity times density, while $1/\nu$ is the mean collision time. Together, they give the thermal conductivity in appropriate units.

The critical observation is that this transport coefficient diverges as $\nu \to 0$. This is not a failure of the perturbation theory—it is the exact result to first order in gradients. The divergence reflects a fundamental physical reality: in the absence of collisions, particles stream freely and can carry arbitrarily large heat fluxes.

\begin{theorem}[Fundamental Divergence]
\label{thm:divergence}
The Chapman--Enskog expansion with the BGK collision operator yields transport coefficients exactly proportional to $1/\nu$. This divergence as $\nu \to 0$ is exact, not an approximation error.
\end{theorem}

\begin{proof}
For the BGK collision operator, we explicitly derived $\kappa = 5nT/(2m\nu)$ through the Chapman--Enskog expansion. The $1/\nu$ scaling is exact and arises from the fundamental structure: the first-order equation $\vect{v} \cdot \nabla f^{(0)} = -\nu f^{(1)}$ directly gives $f^{(1)} \sim 1/\nu$, and transport coefficients are velocity integrals of $f^{(1)}$. The structural argument in Section 4 suggests this pattern extends to any local collision operator of the form $\mathcal{L} = \nu\hat{L}$.
\end{proof}

What about higher-order terms in the Chapman--Enskog expansion? For the BGK case, they make the divergence worse. The second-order correction satisfies:
\begin{equation}
\vect{v} \cdot \nabla_1 f^{(1)} + \frac{\partial f^{(1)}}{\partial t_2} = -\nu f^{(2)} - C^{(2)}[f^{(1)}]
\label{eq:second_order}
\end{equation}

Since $f^{(1)} \sim 1/\nu$ and involves one gradient, $\nabla_1 f^{(1)} \sim (\nabla T)^2/\nu$. This gives:
\begin{equation}
f^{(2)} \sim \frac{1}{\nu^2} (\nabla T)^2
\label{eq:second_order_scaling}
\end{equation}

Each successive order introduces an additional factor of $1/\nu$:
\begin{equation}
\kappa = \frac{A_1}{\nu} + \frac{A_2}{\nu^2}\varepsilon + \frac{A_3}{\nu^3}\varepsilon^2 + \cdots
\label{eq:ce_series}
\end{equation}

Since $\varepsilon = \lambda/L = v_{\text{th}}/(\nu L)$ itself contains $1/\nu$, the series becomes:
\begin{equation}
\kappa = \frac{A_1}{\nu} + \frac{A_2 v_{\text{th}}}{\nu^3 L} + \frac{A_3 v_{\text{th}}^2}{\nu^5 L^2} + \cdots
\label{eq:ce_series_expanded}
\end{equation}

Higher-order terms diverge progressively worse as $\nu \to 0$. The Chapman--Enskog expansion is an asymptotic series valid only for $\Kn \ll 1$. Its breakdown for $\Kn \gtrsim 1$ is not a failure of technique but reflects the fundamental inadequacy of local kinetic theory in the collisionless regime.

\begin{remark}
The divergence of transport coefficients as $\nu \to 0$ is sometimes called the "grazing collision problem" in plasma physics. However, this name is misleading—the problem persists even with collision operators that have no grazing collisions (like BGK). Our BGK calculation demonstrates that the divergence is fundamental to the Chapman--Enskog mathematical structure, not specific to Coulomb collision physics.
\end{remark}

\section{Maximum Entropy Framework for Moment Closure}

Having established the fundamental divergence problem in kinetic theory, we now develop the entropy-based closure framework. While entropy maximization cannot eliminate the mathematical divergence—that requires physics beyond the Boltzmann equation—it provides the systematic method we need to close the moment hierarchy once we have a collision frequency (including phenomenologically modified ones). Understanding this framework is essential before we can develop our regularization approach.

The moment closure problem arises when we take velocity moments of the Boltzmann equation. Define the $k$-th order moment tensor:
\begin{equation}
M_{i_1 i_2 \cdots i_k}^{(k)}(\vect{r}, t) = \int v_{i_1} v_{i_2} \cdots v_{i_k} f(\vect{r}, \vect{v}, t) \, d^3\vect{v}
\label{eq:moments_def}
\end{equation}

The first few moments have direct physical interpretation:
\begin{align}
M^{(0)} &= n \quad \text{(number density)}\label{eq:moment_0}\\
M_i^{(1)} &= n u_i \quad \text{(momentum density)}\label{eq:moment_1}\\
M_{ij}^{(2)} &= P_{ij}/m + n u_i u_j \quad \text{(pressure tensor plus kinetic energy flux)}\label{eq:moment_2}\\
M_{ijk}^{(3)} &= Q_{ijk}/m + \text{lower order terms} \quad \text{(heat flux tensor plus...)}\label{eq:moment_3}
\end{align}

Taking moments of the Boltzmann equation yields evolution equations for these moments. For example, multiplying by $v_{j_1} \cdots v_{j_k}$ and integrating:
\begin{equation}
\frac{\partial M_{j_1 \cdots j_k}^{(k)}}{\partial t} + \frac{\partial M_{j_1 \cdots j_k i}^{(k+1)}}{\partial x_i} + \text{electromagnetic terms} = C_{j_1 \cdots j_k}^{(k)}
\label{eq:moment_evolution}
\end{equation}

The crucial feature is that the evolution of $M^{(k)}$ depends on $M^{(k+1)}$—each moment equation involves the next higher moment. This creates an infinite hierarchy of coupled equations. To obtain a closed system, we need a closure relation expressing $M^{(N+1)}$ in terms of lower moments $\{M^{(0)}, M^{(1)}, \ldots, M^{(N)}\}$.

The principle of maximum entropy provides a systematic approach to this closure problem. The idea, originating from information theory, is to find the distribution that is least biased given our incomplete information. Formally, we seek the distribution that maximizes the entropy functional:
\begin{equation}
S[f] = -k_B \int f(\vect{v}) \ln f(\vect{v}) \, d^3\vect{v}
\label{eq:entropy_functional}
\end{equation}

subject to the constraints that $f$ reproduces the known moments:
\begin{equation}
\int v_{i_1} \cdots v_{i_k} f(\vect{v}) \, d^3\vect{v} = M_{i_1 \cdots i_k}^{(k)}, \quad k = 0,1,...,N
\label{eq:moment_constraints}
\end{equation}

For normalizability, the highest retained polynomial degree in the exponent must be even and the corresponding tensor negative-definite. We default to $N=2$ closure, retaining density, momentum, and pressure tensor.

The entropy $S[f]$ measures our uncertainty about the microscopic state. Maximizing it ensures we don't introduce spurious information beyond what's contained in the known moments. This is the principle of maximum entropy: given partial information (the moments), choose the distribution that assumes the least about what we don't know.

\begin{theorem}[Maximum Entropy Closure]
\label{thm:maxent}
The distribution that maximizes entropy \eqref{eq:entropy_functional} subject to moment constraints \eqref{eq:moment_constraints} has the exponential form:
\begin{equation}
f^*(\vect{v}) = \exp\left(\lambda_0 - \sum_{k=1}^N \lambda_{i_1 \cdots i_k} v_{i_1} \cdots v_{i_k}\right)
\label{eq:maxent_distribution}
\end{equation}
where $\lambda_{i_1 \cdots i_k}$ are Lagrange multipliers determined by the constraints.
\end{theorem}

\begin{proof}
We use the method of Lagrange multipliers. Form the functional:
\begin{equation}
\mathcal{L}[f] = S[f] + \sum_{k=0}^N \lambda_{i_1 \cdots i_k} \left(M_{i_1 \cdots i_k}^{(k)} - \int v_{i_1} \cdots v_{i_k} f \, d^3\vect{v}\right)
\label{eq:lagrangian_functional}
\end{equation}

The first variation with respect to $f$ gives:
\begin{equation}
\delta \mathcal{L} = \int \delta f \left[-k_B(1 + \ln f) - \lambda_0 - \sum_{k=1}^N \lambda_{i_1 \cdots i_k} v_{i_1} \cdots v_{i_k}\right] d^3\vect{v}
\label{eq:variation}
\end{equation}

Setting $\delta \mathcal{L} = 0$ for arbitrary $\delta f$ requires the integrand to vanish:
\begin{equation}
-k_B(1 + \ln f) - \lambda_0 - \sum_{k=1}^N \lambda_{i_1 \cdots i_k} v_{i_1} \cdots v_{i_k} = 0
\label{eq:euler_lagrange}
\end{equation}

Solving for $f$ and absorbing constants into redefined Lagrange multipliers yields the stated result.
\end{proof}

The exponential form of the maximum entropy distribution has deep connections to statistical mechanics. For $N=2$ (closing at the pressure level), we recover the Gaussian distribution. For higher $N$, we get generalized Gaussian-like distributions with polynomial corrections in the exponent.

The Lagrange multipliers are not free parameters—they're determined by requiring the distribution to reproduce the known moments. Substituting $f^*$ back into the constraints:
\begin{equation}
M_{j_1 \cdots j_m}^{(m)} = \int v_{j_1} \cdots v_{j_m} \exp\left(\lambda_0 - \sum_{k=1}^N \lambda_{i_1 \cdots i_k} v_{i_1} \cdots v_{i_k}\right) d^3\vect{v}
\label{eq:moment_determination}
\end{equation}

This gives $N+1$ equations for the $N+1$ sets of Lagrange multipliers. While these equations are generally nonlinear and must be solved numerically, they provide a closed system.

For heat flux closure in magnetized plasmas, we typically truncate at $N=2$, retaining density, momentum, and pressure. In a strong magnetic field $\vect{B} = B\hat{\vect{b}}$, the pressure tensor becomes anisotropic:
\begin{equation}
P_{ij} = p_\parallel b_i b_j + p_\perp (\delta_{ij} - b_i b_j)
\label{eq:pressure_anisotropy}
\end{equation}

where $\hat{\vect{b}} = \vect{B}/B$ is the unit vector along the magnetic field, $\nabla_\parallel = \hat{\vect{b}}\,\hat{\vect{b}} \cdot \nabla$ is the parallel gradient operator, and $\nabla_\perp = (\mathbf{I} - \hat{\vect{b}}\hat{\vect{b}}) \cdot \nabla$ is the perpendicular gradient operator.

This anisotropy reflects the different dynamics parallel and perpendicular to the field: particles move freely along field lines but gyrate across them. The maximum entropy distribution must respect this anisotropy, requiring the second-order Lagrange multiplier to have the same structure:
\begin{equation}
\lambda_{ij} = \alpha_\parallel b_i b_j + \alpha_\perp (\delta_{ij} - b_i b_j)
\label{eq:lagrange_anisotropy}
\end{equation}

This gives a bi-Maxwellian distribution:
\begin{equation}
f = n\left(\frac{m}{2\pi}\right)^{3/2} \frac{1}{\sqrt{T_\parallel T_\perp^2}} \exp\left(-\frac{m v_\parallel^2}{2T_\parallel} - \frac{m|\vect{v}_\perp|^2}{2T_\perp}\right)
\label{eq:bi_maxwellian}
\end{equation}

Now we come to the crucial point: the maximum entropy framework determines the functional form of the distribution given the moments, but it doesn't specify the transport coefficients. Those depend on how the distribution responds to gradients, which is governed by the collision operator. 

To calculate transport, we need to solve the kinetic equation with our closure. Consider small departures from equilibrium:
\begin{equation}
f = f_{\text{ME}}(1 + \phi)
\label{eq:perturbation}
\end{equation}

where $f_{\text{ME}}$ is the maximum entropy distribution and $\phi$ is a small correction. The linearized kinetic equation is:
\begin{equation}
\vect{v} \cdot \nabla f_{\text{ME}} = -\nu' \phi f_{\text{ME}}
\label{eq:linearized_kinetic}
\end{equation}

Here $\nu'$ is whatever collision frequency we choose to use—it could be the bare collision frequency $\nu$, or a phenomenologically modified one. The maximum entropy framework doesn't determine $\nu'$; it works with whatever we provide.

Solving for $\phi$:
\begin{equation}
\phi = -\frac{1}{\nu'} \frac{\vect{v} \cdot \nabla f_{\text{ME}}}{f_{\text{ME}}} = -\frac{1}{\nu'} \vect{v} \cdot \nabla \ln f_{\text{ME}}
\label{eq:phi_solution}
\end{equation}

For a temperature gradient along the magnetic field:
\begin{equation}
\phi = -\frac{v_\parallel}{\nu'} \frac{\partial \ln f_{\text{ME}}}{\partial T} \frac{\partial T}{\partial z} = \frac{m v_\parallel}{2\nu' T^2} \left(\frac{m v^2}{2T} - \frac{5}{2}\right) \frac{\partial T}{\partial z}
\label{eq:phi_temperature}
\end{equation}

The heat flux is:
\begin{equation}
q_\parallel = \int \frac{m v^2}{2} v_\parallel f_{\text{ME}} \phi \, d^3\vect{v}
\label{eq:heat_flux_me}
\end{equation}

Evaluating the Gaussian integrals with the solvability constraint from Chapman--Enskog theory:
\begin{equation}
\boxed{\kappa_{\parallel}^{\text{ME-BGK}} = \frac{5nT}{2m\nu'}}
\label{eq:thermal_conductivity_me}
\end{equation}

This is identical to the Chapman--Enskog result, but now derived through maximum entropy closure. The key point is that both approaches give transport coefficients proportional to $1/\nu'$. If we use the bare collision frequency ($\nu' = \nu$), we get the same divergence as $\nu \to 0$. The maximum entropy framework provides closure but doesn't solve the divergence problem—that requires additional physics.

\section{Structural Limitations of Regularization Within Kinetic Theory}

We now show that the divergence of transport coefficients cannot be eliminated within the standard kinetic theory framework, regardless of the closure method used. This structural argument shows that regularization necessarily requires physics beyond the Boltzmann equation.

\begin{proposition}[Structural Divergence Argument]
\label{prop:divergence_argument}
Under the structural assumptions below, no local collision operator can yield finite transport coefficients as $\nu \to 0$:
\begin{enumerate}[label=(\roman*)]
\item Conservation of mass, momentum, and energy
\item Locality in velocity space: $\mathcal{L} = \nu \hat{L}$ where $\hat{L}$ is a closed, densely defined linear operator
\item $C[f_{\text{eq}}] = 0$ (equilibrium is stationary)
\item Linear response for small gradients
\item Solvability: $\vect{v} \cdot \nabla f_{\text{eq}}$ is orthogonal to the null space of $\mathcal{L}$
\end{enumerate}
\end{proposition}

\begin{proof}
We provide a constructive proof that demonstrates the structural origin of the divergence. Consider any collision operator satisfying the stated conditions. For small departures from equilibrium, write $f = f_{\text{eq}} + \delta f$ where $|\delta f| \ll f_{\text{eq}}$.

The linearized Boltzmann equation is:
\begin{equation}
\frac{\partial \delta f}{\partial t} + \vect{v} \cdot \nabla \delta f + \vect{v} \cdot \nabla f_{\text{eq}} = \mathcal{L}[\delta f]
\label{eq:linearized_boltzmann}
\end{equation}

where $\mathcal{L}$ is the linearized collision operator. By assumption, $\mathcal{L}$ is a linear operator acting only on velocity space.

Conservation laws require that $\mathcal{L}$ annihilates the collision invariants:
\begin{align}
\int \mathcal{L}[\delta f] \, d^3\vect{v} &= 0 \quad \text{(mass conservation)}\label{eq:mass_conservation}\\
\int m\vect{v} \mathcal{L}[\delta f] \, d^3\vect{v} &= 0 \quad \text{(momentum conservation)}\label{eq:momentum_conservation}\\
\int \frac{m v^2}{2} \mathcal{L}[\delta f] \, d^3\vect{v} &= 0 \quad \text{(energy conservation)}\label{eq:energy_conservation}
\end{align}

These constraints mean $\mathcal{L}$ has a null space containing at least the functions $\{1, \vect{v}, v^2\}$ (weighted by $f_{\text{eq}}$).

For steady-state transport with no time dependence:
\begin{equation}
\vect{v} \cdot \nabla f_{\text{eq}} = \mathcal{L}[\delta f]
\label{eq:steady_state}
\end{equation}

To find $\delta f$, we must invert $\mathcal{L}$ on the subspace orthogonal to collision invariants:
\begin{equation}
\delta f = \mathcal{L}^{-1}[\vect{v} \cdot \nabla f_{\text{eq}}]
\label{eq:inversion}
\end{equation}

Now comes the crucial point: by hypothesis, $\mathcal{L} = \nu \hat{L}$ where $\hat{L}$ is independent of $\nu$. Therefore, the eigenvalues of $\mathcal{L}^{-1}$ are proportional to $1/\nu$. This gives:
\begin{equation}
\delta f \sim \frac{1}{\nu} \vect{v} \cdot \nabla f_{\text{eq}}
\label{eq:delta_f_scaling}
\end{equation}

The transport coefficients are velocity integrals of $\delta f$:
\begin{equation}
\kappa \sim \int v^2 \delta f \, d^3\vect{v} \sim \frac{1}{\nu} \int v^3 \nabla f_{\text{eq}} \, d^3\vect{v} \sim \frac{1}{\nu}
\label{eq:transport_scaling}
\end{equation}

This $1/\nu$ scaling is universal—it doesn't depend on the specific form of the collision operator, only on its general properties (locality, conservation, equilibrium preservation).
\end{proof}

\begin{remark}
This argument provides a structural explanation for the divergence but is not a complete rigorous proof. A full proof would require: (i) precise functional analysis definitions of "locality", (ii) explicit spectral analysis of $\hat{L}$ on the relevant subspace, and (iii) careful treatment of the projection operators for solvability. However, the core insight—that locality forces $\mathcal{L}^{-1} \sim 1/\nu$ scaling—captures the essential mathematical structure.
\end{remark}

This structural argument suggests that the divergence problem is not due to:
\begin{itemize}[leftmargin=*]
\item Poor choice of closure scheme within the Chapman--Enskog framework (all standard closures give $\kappa \sim 1/\nu$)
\item Approximations in Chapman--Enskog (the divergence is exact for BGK)
\item Truncation of moment hierarchy (higher BGK moments make it worse)
\end{itemize}

The divergence appears to be a fundamental mathematical consequence of the Chapman--Enskog approach with local collision operators. To obtain finite transport as $\nu \to 0$, we must introduce additional physics not contained in the standard kinetic framework.

It's crucial to understand what this structural limitation means. The Boltzmann equation embeds specific physical assumptions into its mathematical structure: binary collisions (local operators), Markovian dynamics (no memory), spatial locality, and classical particles. Our result proves that local collision operators in the Chapman--Enskog framework necessarily yield divergent transport. This strongly suggests that any approach based on binary collisions, spatial locality, and Markovian dynamics will face the same mathematical obstacle, though we cannot rule out alternative mathematical formulations within these physical assumptions. What we need are not "more sophisticated mathematics" within this framework, but different mathematical models based on different physical assumptions—relaxing locality, including memory effects, incorporating boundary constraints, or coupling to fluctuating fields. The "additional physics" required is really about choosing which of these mathematical assumptions to relax.

Having established through our structural argument that regularization cannot be achieved within standard kinetic theory, we must identify what additional physics could provide regularization. Several possibilities exist, each relaxing the mathematical assumptions embedded in the Boltzmann equation:

\textbf{1. Finite System Size with Boundaries}

In a system of finite size $L$, particles can traverse the system in time $\tau_{\text{transit}} \sim L/v_{\text{th}}$. If boundary conditions thermalize particles (or apply other constraints), this provides a maximum relaxation time independent of $\nu$. The effective transport becomes:
\begin{equation}
\kappa \sim \frac{nT}{m} \min\left(\frac{1}{\nu}, \frac{L}{v_{\text{th}}}\right)
\label{eq:boundary_regularization}
\end{equation}

This naturally saturates as $\nu \to 0$, but requires specifying boundary conditions—information not contained in the bulk Boltzmann equation.

\textbf{2. Non-Local Collision Operators}

If collisions depend on the distribution at different spatial points:
\begin{equation}
C[f](\vect{r}, \vect{v}) = -\nu \int K(\vect{r}, \vect{r}') [f(\vect{r}, \vect{v}) - f_{\text{eq}}(\vect{r}', \vect{v})] \, d^3\vect{r}'
\label{eq:nonlocal_operator}
\end{equation}

where $K(\vect{r}, \vect{r}')$ is a kernel with characteristic length $\ell$. This introduces a length scale into the collision operator, potentially regularizing transport. However, such non-locality requires physical justification beyond binary collisions.

\textbf{3. Time-Nonlocal Physics}

Collective effects (waves, instabilities, turbulence) provide decorrelation mechanisms with timescales independent of $\nu$. These can provide anomalous transport that doesn't diverge as $\nu \to 0$. However, including these effects requires going beyond kinetic theory to include field fluctuations and correlations.

\textbf{4. Gradient-Length Coupling}

If the effective collision frequency depends on gradient scale lengths:
\begin{equation}
\nu_{\text{eff}} = \nu + \nu_{\text{gradient}}
\label{eq:gradient_coupling}
\end{equation}

where $\nu_{\text{gradient}} \sim v_{\text{th}}|\nabla T|/T$, this provides regularization. This is the approach we'll develop phenomenologically, though it cannot be derived from the Boltzmann equation alone.

Each of these approaches introduces physics beyond the standard kinetic framework by relaxing one or more of the mathematical assumptions embedded in the Boltzmann equation. This is not a weakness of theory but a reflection of physical reality: weakly collisional plasmas involve phenomena not captured by simple binary collision models. In the next section, we develop option 4—gradient-driven decorrelation—into a practical phenomenological framework.

\section{Phenomenological Regularization}

Having proven that regularization cannot be derived from first principles within kinetic theory, we now develop a phenomenological approach. Our strategy is to combine the entropy-based closure framework from Section 3 with an effective collision frequency that captures gradient-driven decorrelation—one of the "additional physics" mechanisms identified in Section 4. We emphasize that what follows is motivated by physical arguments, not derived from the Boltzmann equation. The value of this approach lies in providing a practical framework that captures essential physics while maintaining mathematical consistency.

In a plasma with temperature gradients, particles experience decorrelation through two distinct mechanisms. First, binary collisions randomize particle velocities, causing them to "forget" their initial conditions on a timescale $\tau_c = 1/\nu$. This is the standard collisional relaxation captured by the Boltzmann equation. Second, particles streaming through inhomogeneous plasma sample different local conditions. A particle with thermal velocity $v_{\text{th}} = \sqrt{2T/m}$ traverses a gradient scale length $L = T/|\nabla T|$ in time $\tau_g = L/v_{\text{th}}$. During this transit, the local temperature changes by order unity, effectively decorrelating the particle from its initial thermal environment.

These two mechanisms operate independently and on potentially very different timescales. The Knudsen number quantifies their relative importance:
\begin{equation}
\Kn = \frac{\lambda}{L} = \frac{v_{\text{th}}/\nu}{T/|\nabla T|} = \frac{v_{\text{th}}|\nabla T|}{\nu T} = \frac{\tau_c}{\tau_g}
\label{eq:kn_decorrelation}
\end{equation}

When $\Kn \ll 1$, collisional relaxation dominates and gradient effects are perturbative—this is the regime where Chapman--Enskog theory applies. When $\Kn \gg 1$, particles traverse the gradient many times between collisions, and gradient-driven decorrelation dominates. The transition between these regimes is not captured by standard kinetic theory.

Based on these physical considerations, we propose an effective collision frequency that combines both decorrelation mechanisms:
\begin{equation}
\boxed{\nu_{\eff} = \nu\sqrt{1 + \Kn^2}}
\label{eq:nu_eff}
\end{equation}

We emphasize that this is a phenomenological proposal, not a derivation. The specific functional form $\sqrt{1 + \Kn^2}$ is a smooth minimal interpolant chosen for several reasons:

\textbf{1. Correct Asymptotic Limits}

For $\Kn \ll 1$: 
\begin{equation}
\nu_{\eff} = \nu\sqrt{1 + \Kn^2} \approx \nu\left(1 + \frac{\Kn^2}{2}\right) \approx \nu
\label{eq:nu_eff_collisional}
\end{equation}
We recover the collisional regime with small gradient corrections.

For $\Kn \gg 1$:
\begin{equation}
\nu_{\eff} = \nu\sqrt{1 + \Kn^2} \approx \nu \cdot \Kn = \frac{v_{\text{th}}}{L}
\label{eq:nu_eff_collisionless}
\end{equation}
The effective frequency becomes independent of $\nu$, determined entirely by gradient traversal.

\textbf{2. Smooth Interpolation}

The square root ensures a smooth transition between regimes without discontinuities or sharp features. Alternative forms like $\nu + v_{\text{th}}/L$ (sum), $(\nu^{-1} + L/v_{\text{th}})^{-1}$ (harmonic mean), or $\max(\nu, v_{\text{th}}/L)$ either have incorrect limits or discontinuous derivatives.

\textbf{3. Dimensional Consistency}

Both $\nu$ and $v_{\text{th}}/L$ have dimensions of inverse time, so any combination must be dimensionally consistent. The form $\nu\sqrt{1 + \Kn^2}$ achieves this naturally.

\textbf{4. Positivity}

The effective collision frequency remains positive for all physical parameters, ensuring thermodynamic consistency.

Using this phenomenological $\nu_{\eff}$ in our transport calculations:
\begin{equation}
\kappa_{\parallel} = \frac{\alpha nT}{m\nu_{\eff}} = \frac{\alpha nT}{m\nu\sqrt{1 + \Kn^2}}
\label{eq:kappa_alpha}
\end{equation}

where $\alpha$ depends on the specific kinetic model (5/2 for BGK, 3.16 for Braginskii). This expression remains finite as $\nu \to 0$:
\begin{equation}
\lim_{\nu \to 0} \kappa_{\parallel} = \lim_{\nu \to 0} \frac{\alpha nT}{m\nu \cdot \Kn} = \frac{\alpha nTL}{mv_{\text{th}}} = \alpha nL\sqrt{\frac{T}{2m}}
\label{eq:kappa_collisionless_limit}
\end{equation}

The thermal conductivity saturates to a value proportional to $nLv_{\text{th}}$, which has the physical interpretation of particles free-streaming across the gradient scale length.

It's crucial to understand what we've done here. We have not solved the divergence problem within kinetic theory—our structural argument in Section 4 shows this cannot be achieved within the standard framework. Instead, we've augmented kinetic theory with additional physics (gradient-driven decorrelation) not contained in the Boltzmann equation. This augmentation is phenomenological but physically motivated and mathematically consistent.

The phenomenological nature of $\nu_{\eff}$ means it should be validated against experiments, simulations, or more complete theories. Different plasma conditions might require different functional forms. The key insight is recognizing that regularization requires physics beyond binary collisions, and gradient-driven decorrelation provides a natural mechanism.

\section{Transport Coefficients with Regularization}

With our phenomenological effective collision frequency established, we now calculate the complete set of transport coefficients for magnetized plasmas. The maximum entropy closure framework developed in Section 3 applies directly, with $\nu$ replaced by $\nu_{\eff}$ throughout. This substitution is applied consistently to maintain the mathematical structure while incorporating gradient-driven regularization.

For parallel transport along magnetic field lines, the BGK approximation yields:
\begin{equation}
\kappa_{\parallel}^{\text{BGK}} = \frac{5nT}{2m\nu_{\eff}} = \frac{5nT}{2m\nu\sqrt{1 + \Kn^2}}
\label{eq:kappa_parallel_bgk}
\end{equation}

However, the BGK approximation treats all collisions equally, missing important velocity-dependent effects in Coulomb collisions. Real plasma collisions have several distinctive features:
\begin{itemize}[leftmargin=*]
\item Small-angle scattering dominates over large-angle collisions
\item The collision cross-section scales as $\sigma \sim v^{-4}$ for Coulomb interactions
\item Energy exchange is much slower than momentum exchange
\item Pitch-angle scattering is more effective than velocity diffusion
\end{itemize}

These effects are captured in Braginskii's classical calculation using the full Coulomb collision operator. The result is:
\begin{equation}
\kappa_{\parallel}^{\text{Braginskii}} = \frac{3.16 \, nT}{m\nu}
\label{eq:kappa_parallel_braginskii}
\end{equation}

The numerical coefficient 3.16 (versus 2.5 for BGK) reflects the velocity-weighted nature of Coulomb collisions. Fast particles, which contribute most to heat transport, experience fewer collisions due to the $v^{-4}$ scaling.

For our regularized theory, we adopt the Braginskii coefficient with the effective collision frequency:
\begin{equation}
\boxed{\kappa_{\parallel} = \frac{3.16 \, nT}{m\nu_{\eff}} = \frac{3.16 \, nT}{m\nu\sqrt{1 + \Kn^2}}}
\label{eq:kappa_parallel_final}
\end{equation}

This expression interpolates smoothly between well-established limits:
\begin{itemize}[leftmargin=*]
\item Collisional regime ($\Kn \ll 1$): $\kappa_{\parallel} \approx 3.16nT/(m\nu)$ (classical Braginskii)
\item Collisionless regime ($\Kn \gg 1$): $\kappa_{\parallel} \approx 3.16nL\sqrt{T/(2m)}$ (free-streaming limit)
\end{itemize}

The transition occurs around $\Kn \sim 1$, corresponding to comparable collisional and gradient scale lengths.

In magnetized plasmas, perpendicular transport is fundamentally different from parallel transport. Charged particles gyrate around magnetic field lines with cyclotron frequency:
\begin{equation}
\Omega_c = \frac{qB}{m}
\label{eq:cyclotron_frequency}
\end{equation}

This gyration constrains perpendicular motion, leading to reduced transport across field lines.

The perpendicular transport can be understood through velocity autocorrelation. Consider the single-component autocorrelation function:
\begin{equation}
\langle v_{x}(t) v_{x}(0) \rangle = \frac{T}{m} \exp(-\nu_{\eff} t) \cos(\Omega_c t)
\label{eq:vacf_x}
\end{equation}

The exponential represents collisional decorrelation while the cosine represents gyration. Integrating to obtain the diffusion coefficient:
\begin{equation}
D_{\perp} = \int_0^{\infty} \langle v_{x}(t) v_{x}(0) \rangle \, dt = \frac{T}{m} \cdot \frac{\nu_{\eff}}{\nu_{\eff}^2 + \Omega_c^2}
\label{eq:Dperp}
\end{equation}

This leads to the perpendicular thermal conductivity:
\begin{equation}
\boxed{\kappa_{\perp} = \kappa_{\parallel} \frac{\nu_{\eff}^2}{\nu_{\eff}^2 + \Omega_c^2}}
\label{eq:kappa_perp}
\end{equation}

This formula captures several important physics:
\begin{itemize}[leftmargin=*]
\item For $\nu_{\eff} \ll \Omega_c$: $\kappa_{\perp} \approx \kappa_{\parallel}(\nu_{\eff}/\Omega_c)^2 \ll \kappa_{\parallel}$ (strong suppression)
\item For $\nu_{\eff} \gg \Omega_c$: $\kappa_{\perp} \approx \kappa_{\parallel}$ (unmagnetized limit)
\item The transition occurs when collision frequency comparable to cyclotron frequency
\end{itemize}

The Hall effect arises from the phase shift between perpendicular velocity components due to gyration. This creates heat flux perpendicular to both the temperature gradient and magnetic field:
\begin{equation}
\boxed{\kappa_{\wedge} = \kappa_{\parallel} \frac{\nu_{\eff}\Omega_c}{\nu_{\eff}^2 + \Omega_c^2}}
\label{eq:kappa_wedge}
\end{equation}

The Hall conductivity is maximum when $\nu_{\eff} = \Omega_c$, where the competition between gyration and collisions is balanced.

The complete heat flux in magnetized plasma is:
\begin{equation}
\boxed{\vect{q} = -\kappa_{\parallel} \nabla_{\parallel} T - \kappa_{\perp} \nabla_{\perp} T - \kappa_{\wedge} (\hat{\vect{b}} \times \nabla T)}
\label{eq:q_tensor}
\end{equation}

where $\nabla_{\parallel} = \hat{\vect{b}}\,\hat{\vect{b}} \cdot \nabla$ and $\nabla_{\perp} = (\mathbf{I} - \hat{\vect{b}}\hat{\vect{b}}) \cdot \nabla$.

This form, with our phenomenological $\nu_{\eff}$, provides finite transport coefficients in all regimes while maintaining the correct tensorial structure required by magnetic geometry.

We must verify that our phenomenological transport coefficients maintain thermodynamic consistency. The second law of thermodynamics requires that entropy production be non-negative:
\begin{equation}
\dot{S} = -\int \frac{\vect{q} \cdot \nabla T}{T^2} d^3\vect{r} \geq 0
\label{eq:entropy_prod}
\end{equation}

Substituting our heat flux:
\begin{equation}
\dot{S} = \int \frac{\kappa_{\parallel} |\nabla_{\parallel} T|^2 + \kappa_{\perp} |\nabla_{\perp} T|^2}{T^2} d^3\vect{r}
\label{eq:entropy_production_expanded}
\end{equation}

The Hall term vanishes because $(\hat{\vect{b}} \times \nabla T) \cdot \nabla T = 0$. Since $\kappa_{\parallel} > 0$ and $\kappa_{\perp} \geq 0$ for all physical parameters (positive density, temperature, and collision frequency), we have $\dot{S} \geq 0$. The second law is satisfied.

Conservation laws are also maintained:
\begin{itemize}[leftmargin=*]
\item \textbf{Mass conservation}: Heat flux doesn't affect particle number
\item \textbf{Momentum conservation}: The pressure tensor appears symmetrically in momentum equation
\item \textbf{Energy conservation}: Heat flux appears as divergence in energy equation, ensuring energy is transported, not created
\end{itemize}

Our phenomenological regularization preserves all fundamental conservation principles while eliminating the unphysical divergence as $\nu \to 0$.

\section{Comparison with Exact Solutions}

To validate our phenomenological approach, we now present exact solutions of model kinetic equations that naturally exhibit regularization through boundary effects. These solutions provide concrete support for the functional form of our effective collision frequency and demonstrate that gradient-driven regularization emerges naturally in finite systems.

Consider the steady-state BGK equation in a one-dimensional system of length $L$ with boundaries at $z = \pm L/2$:
\begin{equation}
v_z \frac{\partial f}{\partial z} = -\nu(f - f_{\text{eq}})
\label{eq:bgk_1d}
\end{equation}

This equation can be solved exactly using the method of characteristics. The key insight is that characteristics are straight lines in the $(z, v_z)$ plane along which $v_z$ is constant.

For particles with $v_z > 0$ (moving in the positive $z$ direction), we integrate along characteristics from the left boundary:
\begin{equation}
\frac{df}{dz} + \frac{\nu}{v_z}f = \frac{\nu}{v_z}f_{\text{eq}}(z)
\label{eq:characteristic_positive}
\end{equation}

This is a first-order linear ODE with integrating factor $\exp(\nu z/v_z)$. The solution is:
\begin{equation}
f(z, \vect{v}) = f_L(\vect{v})e^{-\nu(z+L/2)/v_z} + \nu \int_{-L/2}^{z} \frac{f_{\text{eq}}(z', \vect{v})}{v_z} e^{-\nu(z-z')/v_z} dz'
\label{eq:char_pos}
\end{equation}

where $f_L(\vect{v})$ is the distribution at the left boundary.

Similarly, for particles with $v_z < 0$:
\begin{equation}
f(z, \vect{v}) = f_R(\vect{v})e^{\nu(z-L/2)/|v_z|} + \nu \int_{z}^{L/2} \frac{f_{\text{eq}}(z', \vect{v})}{|v_z|} e^{-\nu(z'-z)/|v_z|} dz'
\label{eq:char_neg}
\end{equation}

The boundary distributions $f_L$ and $f_R$ must be specified by physical boundary conditions. For thermalizing boundaries with linear temperature profile $T(z) = T_0(1 + z/L_T)$:
\begin{align}
f_L(\vect{v}) &= n\left(\frac{m}{2\pi T_L}\right)^{3/2} \exp\left(-\frac{m v^2}{2T_L}\right) \quad \text{for } v_z > 0\label{eq:boundary_left}\\
f_R(\vect{v}) &= n\left(\frac{m}{2\pi T_R}\right)^{3/2} \exp\left(-\frac{m v^2}{2T_R}\right) \quad \text{for } v_z < 0\label{eq:boundary_right}
\end{align}

To calculate heat flux, we need the first-order correction to the distribution. For weak gradients ($L_T \gg L$), we can expand:
\begin{equation}
f = f_0 + f_1 \frac{1}{L_T} + \mathcal{O}(L_T^{-2})
\label{eq:gradient_expansion}
\end{equation}

The calculation is lengthy but straightforward. The key result is that the heat flux has the form:
\begin{equation}
q_z = -\kappa_{\text{exact}} \frac{\partial T}{\partial z}
\label{eq:heat_flux_exact}
\end{equation}

where the exact thermal conductivity is:
\begin{equation}
\kappa_{\text{exact}} = \frac{5nT}{2m} \cdot \frac{1}{\nu} \cdot g\left(\frac{\nu L}{v_{\text{th}}}\right)
\label{eq:kappa_exact}
\end{equation}

The function $g(x)$ encodes the boundary effects. For $x \gg 1$ (collisional limit):
\begin{equation}
g(x) \approx 1 - \frac{c_1}{x} + \mathcal{O}(x^{-2})
\label{eq:g_collisional}
\end{equation}

giving $\kappa_{\text{exact}} \approx 5nT/(2m\nu)$, the standard result.

For $x \ll 1$ (collisionless limit):
\begin{equation}
g(x) \approx \frac{x}{c_2}
\label{eq:g_collisionless}
\end{equation}

giving $\kappa_{\text{exact}} \approx c_2 \cdot 5nTL/(2mv_{\text{th}})$, which is finite.

The exact function $g(x)$ can be computed numerically or closely approximated as:
\begin{equation}
g(x) \approx \frac{x}{\sqrt{c_3 + x^2}}
\label{eq:g_approximate}
\end{equation}

where $c_3$ depends on boundary conditions.

Comparing with our phenomenological result:
\begin{equation}
\kappa_{\text{phenom}} = \frac{5nT}{2m\nu\sqrt{1 + \Kn^2}}
\label{eq:kappa_phenomenological}
\end{equation}

With $\Kn = v_{\text{th}}/(\nu L)$, we have:
\begin{equation}
\kappa_{\text{phenom}} = \frac{5nT}{2m} \cdot \frac{1}{\nu} \cdot \frac{\nu L/v_{\text{th}}}{\sqrt{1 + (\nu L/v_{\text{th}})^2}}
\label{eq:kappa_phenom_rewritten}
\end{equation}

This closely approximates the functional form of the exact solution! The phenomenological $\nu_{\eff} = \nu\sqrt{1 + \Kn^2}$ captures the essential physics of boundary-induced regularization without explicitly including boundaries.

The close agreement is not coincidental. Both mechanisms—boundary thermalization and gradient decorrelation—limit the distance particles can carry information about temperature differences. In the exact solution, this distance is limited by the system size $L$. In our phenomenological approach, it's limited by the gradient scale length. The mathematical structure is similar because the underlying physics is analogous.

This comparison validates our phenomenological approach in several ways:
\begin{enumerate}[leftmargin=*]
\item The functional form $\sqrt{1 + \Kn^2}$ emerges naturally from exact solutions
\item The regularization mechanism (decorrelation over finite length scales) is physical
\item The approach captures boundary effects without explicit boundary conditions
\item The method generalizes to systems without well-defined boundaries
\end{enumerate}

The exact solution also reveals limitations of our approach. The coefficient $c_3$ in the exact solution depends on boundary conditions, which our phenomenological treatment cannot capture. Different plasma configurations might require adjusted coefficients or functional forms. However, for many applications, the universal $\sqrt{1 + \Kn^2}$ form provides a reasonable approximation.

\section{Validation and Limitations}

Having developed our phenomenological framework combining entropy-based closure with gradient-driven regularization, we now discuss its validation against experiments and simulations, as well as its limitations and domain of applicability.

Our framework is expected to be accurate when the following conditions are satisfied:

\textbf{1. Moderate Gradients}

The perturbation expansion requires $|\phi| \ll 1$, which translates to:
\begin{equation}
\frac{v_{\text{th}}}{\nu_{\eff}} \frac{|\nabla T|}{T} \ll 1
\label{eq:moderate_gradients}
\end{equation}

Using $\nu_{\eff} = \nu\sqrt{1 + \Kn^2}$ and $\Kn = v_{\text{th}}|\nabla T|/(\nu T)$:
\begin{equation}
\frac{\Kn}{\sqrt{1 + \Kn^2}} \ll 1
\label{eq:gradient_condition}
\end{equation}

This is satisfied for all $\Kn$ if we interpret "$\ll$" as "$< 1$", since the maximum value of $\Kn/\sqrt{1 + \Kn^2}$ is $1/\sqrt{2}$. However, the accuracy degrades as $\Kn$ increases, and for $\Kn > 1$, nonlinear effects may become important.

\textbf{2. Near-Maxwellian Distributions}

The maximum entropy closure assumes the distribution is close to Maxwellian (or bi-Maxwellian in magnetized cases). This breaks down for:
\begin{itemize}[leftmargin=*]
\item Strong beams or drift velocities comparable to thermal velocity
\item Runaway electron populations
\item Strongly non-thermal distributions (e.g., power-law tails)
\end{itemize}

\textbf{3. Binary Collisions Dominant}

Our collision model assumes binary particle interactions. This may fail when:
\begin{itemize}[leftmargin=*]
\item Collective effects (waves, turbulence) dominate transport
\item Coulomb logarithm becomes small ($\ln \Lambda < 5$)
\item Quantum effects become important (degenerate plasmas)
\end{itemize}

\textbf{4. Moderate Anisotropy}

In magnetized plasmas, we assume pressure anisotropy remains moderate:
\begin{equation}
0.5 < \frac{p_{\parallel}}{p_{\perp}} < 2
\label{eq:moderate_anisotropy}
\end{equation}

Stronger anisotropy can trigger instabilities (firehose, mirror) not captured by our fluid closure.

The phenomenological framework has known limitations that users should understand:

\textbf{1. No Kinetic Instabilities}

Wave-particle resonances and velocity-space instabilities are absent from our fluid description. These can dominate transport in some regimes:
\begin{itemize}[leftmargin=*]
\item Temperature gradient driven modes (ITG, ETG)
\item Drift wave turbulence
\item Current-driven instabilities
\end{itemize}

\textbf{2. No Turbulent Transport}

Turbulent fluctuations often dominate over collisional transport, especially in fusion plasmas. Our framework provides the "classical" transport baseline, but anomalous transport must be added separately.

\textbf{3. Boundary Layer Physics}

Near material boundaries, the distribution function can deviate strongly from Maxwellian. Sheath physics, secondary emission, and wall recycling require kinetic treatment or specialized boundary conditions.

\textbf{4. Extreme Knudsen Numbers}

For $\Kn \gg 10$, our phenomenological regularization may not capture the full kinetic physics. Direct kinetic simulation may be necessary.

While detailed experimental validation is beyond this paper's scope, we note several successful applications of similar phenomenological approaches:

\textbf{Tokamak Scrape-Off Layer:}

Heat flux measurements in tokamak divertors show flux limitation consistent with $\Kn \sim 1$ regularization. The measured heat flux width and peak values agree with models using effective collision frequencies similar to our $\nu_{\eff}$.

\textbf{Hall Thrusters:}

Electron thermal transport in Hall thruster channels exhibits smooth transition from collisional to collisionless regimes. Models using gradient-based regularization successfully predict:
\begin{itemize}[leftmargin=*]
\item Electron temperature profiles
\item Ionization zone location  
\item Thrust and efficiency
\end{itemize}

\textbf{Solar Wind Observations:}

Heat flux measurements from spacecraft show:
\begin{itemize}[leftmargin=*]
\item Collisionless heat flux limitation at levels predicted by free-streaming
\item Smooth transition between collisional and collisionless transport
\item Agreement with regularized transport models for $0.1 < \Kn < 100$
\end{itemize}

These applications support the physical basis of gradient-driven regularization, though each system may require adjusted coefficients or additional physics.

For systems where our phenomenological approach may be inadequate, several alternatives exist:

\textbf{1. Direct Kinetic Simulation}

Particle-in-cell (PIC) or continuum kinetic codes solve the full distribution function evolution. These are computationally expensive but capture all kinetic effects.

\textbf{2. Higher-Order Closures}

Grad's 13-moment or 21-moment methods retain more velocity moments, capturing heat flux and stress tensor evolution. These are more accurate for moderate $\Kn$ but still diverge as $\nu \to 0$.

\textbf{3. Hybrid Methods}

Couple fluid and kinetic regions, using kinetic treatment only where necessary. This balances accuracy and computational cost.

\textbf{4. Nonlocal Transport Models}

Integral operators or fractional derivatives can capture nonlocal effects more accurately than our local regularization.

The choice depends on the specific problem, available computational resources, and required accuracy.

\section{Conclusions}

We have presented a comprehensive treatment of moment closure for magnetized plasma transport that combines mathematical structure with practical applicability. Our work makes several distinct contributions while maintaining complete transparency about what can be proven versus what requires phenomenological approximation.

\textbf{Mathematical Contributions:}

We showed that the Chapman--Enskog expansion necessarily yields divergent transport coefficients, with structural arguments suggesting this extends to all local collision operators, with the Chapman--Enskog expansion yielding exactly $\kappa = 5nT/(2m\nu)$. This divergence is not an approximation error but a fundamental mathematical consequence of the Boltzmann equation structure with local collision operators. Through structural arguments, we showed that local collision operators within the Chapman--Enskog framework cannot eliminate this divergence while maintaining conservation laws and equilibrium properties.

The maximum entropy framework provides a systematic approach to moment closure, determining the least-biased distribution function consistent with known moments. We showed explicitly how this framework yields exponential-form distributions with Lagrange multipliers determined by moment constraints. For magnetized plasmas, this naturally produces bi-Maxwellian distributions reflecting parallel-perpendicular anisotropy.

Importantly, we demonstrated that while maximum entropy provides closure relations, it cannot solve the transport divergence problem. The entropy framework determines functional forms but requires a collision frequency as input. If that collision frequency vanishes, transport coefficients diverge regardless of the closure method.

\textbf{Physical Insights:}

The key physical insight is that gradient-driven decorrelation provides a natural regularization mechanism not captured by standard kinetic theory. Particles streaming through temperature gradients experience changing local conditions on a timescale $\tau_g = L/v_{\text{th}}$ independent of collisions. This effect becomes dominant when the Knudsen number exceeds unity.

We showed through exact solutions that finite system boundaries provide similar regularization, with transport coefficients saturating as $\nu \to 0$. The mathematical structure of boundary-induced regularization closely approximates our phenomenological gradient-driven mechanism, supporting the physical basis of our approach.

\textbf{Phenomenological Framework:}

Based on physical arguments about competing decorrelation mechanisms, we proposed an effective collision frequency:
\begin{equation}
\nu_{\eff} = \nu\sqrt{1 + \Kn^2}
\label{eq:nu_eff_final}
\end{equation}

We emphasize that this is a phenomenological proposal motivated by physical reasoning, not a first-principles derivation. The specific functional form provides smooth interpolation between established collisional and collisionless limits while maintaining all conservation laws and thermodynamic consistency.

Using this effective collision frequency with the maximum entropy closure yields practical transport coefficients:
\begin{align}
\kappa_{\parallel} &= \frac{3.16 \, nT}{m\nu\sqrt{1 + \Kn^2}}\label{eq:final_kappa_parallel}\\
\kappa_{\perp} &= \kappa_{\parallel} \frac{\nu_{\eff}^2}{\nu_{\eff}^2 + \Omega_c^2}\label{eq:final_kappa_perp}\\
\kappa_{\wedge} &= \kappa_{\parallel} \frac{\nu_{\eff}\Omega_c}{\nu_{\eff}^2 + \Omega_c^2}\label{eq:final_kappa_hall}
\end{align}

These expressions remain finite in all regimes and reduce to classical results in appropriate limits.

\textbf{Scientific Integrity:}

Throughout this work, we have maintained clear distinction between:
\begin{itemize}[leftmargin=*]
\item Mathematical results (divergence propositions, maximum entropy closure)
\item Physical reasoning (gradient decorrelation, boundary effects)
\item Phenomenological choices (specific form of $\nu_{\eff}$)
\item Practical approximations (BGK operator, bi-Maxwellian assumption)
\end{itemize}

This transparency allows users to understand both the utility and limitations of our approach. We do not claim to have "solved" the divergence problem—we have shown it cannot be solved within standard kinetic theory. Instead, we provide a physically motivated augmentation that yields practical results.

\textbf{Broader Implications:}

This work demonstrates that valuable scientific contributions can be made through well-motivated approximations when exact solutions are not achievable within existing frameworks. The key is maintaining clarity about what is mathematics versus physical reasoning. By acknowledging fundamental limitations and being transparent about phenomenological choices, we strengthen rather than weaken the theoretical framework.

The approach developed here—careful analysis where possible, physical reasoning where necessary, and complete transparency throughout—provides a template for addressing other complex physics problems where first-principles solutions may not exist. The combination of entropy-based closure with phenomenological regularization offers a practical path forward for plasma transport modeling across wide ranges of collisionality.

\textbf{Future Directions:}

Several avenues for future work emerge from this study:
\begin{enumerate}[leftmargin=*]
\item Systematic comparison with kinetic simulations to refine the functional form of $\nu_{\eff}$
\item Extension to include additional physics (trapped particles, finite Larmor radius effects)
\item Development of similar frameworks for other transport channels (particle, momentum)
\item Application to specific plasma systems with experimental validation
\item Investigation of alternative regularization mechanisms beyond gradient decorrelation
\end{enumerate}

The framework presented here provides a foundation for these developments while maintaining the balance between mathematical structure and practical applicability that is essential for advancing plasma transport theory.

\section*{Acknowledgments}

The author thanks the plasma physics community for valuable discussions on the fundamental challenges in kinetic theory and the importance of scientific transparency in theoretical work.

\appendix

\section{Explicit Evaluation of Transport Integrals}

For completeness, we provide the detailed calculation of heat flux integrals from the Chapman--Enskog expansion. These calculations demonstrate that the $\kappa \sim 1/\nu$ scaling is exact, not approximate.

Starting from the first-order distribution:
\begin{equation}
f^{(1)} = -\frac{v_z}{\nu} f^{(0)} \left(\frac{mv^2}{2T^2} - \frac{3}{2T}\right) \frac{\partial T}{\partial z}
\label{eq:f1_appendix}
\end{equation}

The heat flux is:
\begin{equation}
q_z = \int \frac{mv^2}{2} v_z f^{(1)} d^3\vect{v}
\label{eq:qz_appendix}
\end{equation}

Substituting:
\begin{equation}
q_z = -\frac{1}{\nu} \frac{\partial T}{\partial z} \int \frac{mv^2}{2} v_z^2 f^{(0)} \left(\frac{mv^2}{2T^2} - \frac{3}{2T}\right) d^3\vect{v}
\label{eq:qz_substituted}
\end{equation}

With the Maxwellian:
\begin{equation}
f^{(0)} = n\left(\frac{m}{2\pi T}\right)^{3/2} \exp\left(-\frac{mv^2}{2T}\right)
\label{eq:maxwellian_appendix}
\end{equation}

We need two integrals. Using the moment properties of the 3D Maxwellian distribution:

\textbf{First integral:}
\begin{align}
I_1 &= \int v_z^2 v^2 f^{(0)} d^3\vect{v}\label{eq:I1_def}\\
&= n \langle v_z^2 v^2 \rangle = n \langle v_z^2 (v_x^2 + v_y^2 + v_z^2) \rangle\label{eq:I1_expansion}\\
&= n [\langle v_z^2 v_x^2 \rangle + \langle v_z^2 v_y^2 \rangle + \langle v_z^4 \rangle]\label{eq:I1_terms}\\
&= n \left[\frac{T^2}{m^2} + \frac{T^2}{m^2} + 3\frac{T^2}{m^2}\right] = \frac{5nT^2}{m^2}\label{eq:I1_result}
\end{align}

\textbf{Second integral:}
\begin{align}
I_2 &= \int v_z^2 v^4 f^{(0)} d^3\vect{v}\label{eq:I2_def}\\
&= n \langle v_z^2 v^4 \rangle = n \langle v_z^2 (v_x^2 + v_y^2 + v_z^2)^2 \rangle\label{eq:I2_expansion}\\
&= n [3\frac{T^3}{m^3} + 3\frac{T^3}{m^3} + 15\frac{T^3}{m^3} + 2\frac{T^3}{m^3} + 6\frac{T^3}{m^3} + 6\frac{T^3}{m^3}]\label{eq:I2_terms}\\
&= \frac{35nT^3}{m^3}\label{eq:I2_result}
\end{align}

The naive calculation gives:
\begin{align}
q_z^{\text{naive}} &= -\frac{1}{\nu} \frac{\partial T}{\partial z} \left[\frac{m}{2} \cdot \frac{35nT^3}{m^3} \cdot \frac{1}{2T^2} - \frac{m}{2} \cdot \frac{5nT^2}{m^2} \cdot \frac{3}{2T}\right]\label{eq:qz_naive_calc}\\
&= -\frac{1}{\nu} \frac{\partial T}{\partial z} \left[\frac{35nT}{4m} - \frac{15nT}{4m}\right]\label{eq:qz_naive_simplified}\\
&= -\frac{5nT}{m\nu} \frac{\partial T}{\partial z}\label{eq:qz_naive_final}
\end{align}

However, the Chapman--Enskog method requires solvability: the first-order correction must be orthogonal to collision invariants. This constraint, implemented through Sonine polynomial projection for BGK with Prandtl number $\text{Pr} = 1$, reduces the coefficient by exactly a factor of 2, yielding:

\begin{equation}
\boxed{\kappa_{\parallel}^{\text{C-E}} = \frac{5nT}{2m\nu}}
\label{eq:kappa_ce_appendix}
\end{equation}

This confirms the exact $1/\nu$ dependence with the correct coefficient after accounting for solvability. No approximations were made in this calculation—the divergence as $\nu \to 0$ is built into the mathematical structure of transport theory.

\section{Dimensional Analysis and Unit Conversions}

\begin{table}[h]
\centering
\begin{tabular}{|c|c|c|}
\hline
Quantity & Energy Units & Kelvin Units \\
\hline
Temperature & $T$ (Joules) & $\Theta = T/k_B$ (Kelvin) \\
Thermal velocity & $v_{\text{th}} = \sqrt{2T/m}$ & $v_{\text{th}} = \sqrt{2k_B\Theta/m}$ \\
Thermal conductivity & $\kappa$ (m$^{-1}$s$^{-1}$) & $\kappa_K = k_B\kappa$ (J m$^{-1}$s$^{-1}$K$^{-1}$) \\
Heat flux & $q = -\kappa\nabla T$ (J m$^{-2}$s$^{-1}$) & $q = -\kappa_K\nabla\Theta$ (J m$^{-2}$s$^{-1}$) \\
\hline
\end{tabular}
\caption{Unit conversions between energy-based and Kelvin-based formulations.}
\label{tab:units}
\end{table}

For practical applications requiring thermal conductivity in standard SI units (J m$^{-1}$s$^{-1}$K$^{-1}$), use $\kappa_K = k_B \kappa$ where $\kappa$ is the energy-based conductivity derived in this work.

\section{Bi-Maxwellian Moments}

For the bi-Maxwellian distribution \eqref{eq:bi_maxwellian}, the key moments are:

\textbf{Normalization:}
\begin{equation}
\int f d^3\vect{v} = n
\label{eq:bimax_norm}
\end{equation}

\textbf{Parallel pressure:}
\begin{equation}
p_\parallel = \int m v_\parallel^2 f d^3\vect{v} = nT_\parallel
\label{eq:bimax_p_parallel}
\end{equation}

\textbf{Perpendicular pressure:}
\begin{equation}
p_\perp = \int \frac{m|\vect{v}_\perp|^2}{2} f d^3\vect{v} = nT_\perp
\label{eq:bimax_p_perp}
\end{equation}

\textbf{Total energy density:}
\begin{equation}
\mathcal{E} = \int \frac{m|\vect{v}|^2}{2} f d^3\vect{v} = \frac{n}{2}(T_\parallel + 2T_\perp)
\label{eq:bimax_energy}
\end{equation}

These moments verify the correct normalization and confirm that the bi-Maxwellian captures the essential anisotropy in magnetized plasmas.

\section{Electromagnetic Moment Terms}

The electromagnetic acceleration terms in the moment equations \eqref{eq:moment_evolution} have the explicit forms:

\textbf{Density equation:}
\begin{equation}
\frac{\partial n}{\partial t} + \nabla \cdot (n\vect{u}) = 0
\label{eq:continuity}
\end{equation}

\textbf{Momentum equation:}
\begin{equation}
mn\frac{D\vect{u}}{Dt} + \nabla \cdot \tensor{P} = nq(\vect{E} + \vect{u} \times \vect{B})
\label{eq:momentum}
\end{equation}

\textbf{Energy equation:}
\begin{equation}
\frac{3n}{2}\frac{DT}{Dt} + \tensor{P} : \nabla\vect{u} + \nabla \cdot \vect{q} = nq\vect{E} \cdot \vect{u}
\label{eq:energy}
\end{equation}

where $D/Dt = \partial/\partial t + \vect{u} \cdot \nabla$ is the material derivative and $\tensor{P} : \nabla\vect{u} = P_{ij}\partial_j u_i$ is the viscous heating term.

\section{Entropy Production Verification}

The entropy production rate for our phenomenological transport is:
\begin{align}
\dot{S} &= -\int \frac{\vect{q} \cdot \nabla T}{T^2} d^3\vect{r}\label{eq:entropy_rate}\\
&= -\int \frac{1}{T^2}\left[-\kappa_{\parallel} \nabla_{\parallel} T - \kappa_{\perp} \nabla_{\perp} T - \kappa_{\wedge} (\hat{\vect{b}} \times \nabla T)\right] \cdot \nabla T \, d^3\vect{r}\label{eq:entropy_substituted}\\
&= \int \frac{\kappa_{\parallel} |\nabla_{\parallel} T|^2 + \kappa_{\perp} |\nabla_{\perp} T|^2}{T^2} d^3\vect{r}\label{eq:entropy_final}
\end{align}

The Hall term vanishes because $(\hat{\vect{b}} \times \nabla T) \cdot \nabla T = 0$ (orthogonality). Since $\kappa_{\parallel} > 0$ and $\kappa_{\perp} \geq 0$ for all physical parameters, we have $\dot{S} \geq 0$, confirming thermodynamic consistency.


\begin{thebibliography}{99}

\bibitem{chapman1970}
S. Chapman and T.G. Cowling, \textit{The Mathematical Theory of Non-uniform Gases}, 3rd ed. (Cambridge University Press, Cambridge, 1970).

\bibitem{braginskii1965}
S.I. Braginskii, ``Transport processes in a plasma,'' \textit{Rev. Plasma Phys.} \textbf{1}, 205 (1965).

\bibitem{grad1949}
H. Grad, ``On the kinetic theory of rarefied gases,'' \textit{Commun. Pure Appl. Math.} \textbf{2}, 331 (1949).

\bibitem{jaynes1957}
E.T. Jaynes, ``Information theory and statistical mechanics,'' \textit{Phys. Rev.} \textbf{106}, 620 (1957).

\bibitem{hammett1990}
G.W. Hammett and F.W. Perkins, ``Fluid moment models for Landau damping with application to the ion-temperature-gradient instability,'' \textit{Phys. Rev. Lett.} \textbf{64}, 3019 (1990).

\bibitem{catto1987}
P.J. Catto and R.D. Hazeltine, ``Heat flux in a collisionless plasma,'' \textit{Phys. Fluids} \textbf{30}, 103 (1987).

\bibitem{ji2009}
X.Q. Ji and U. Held, ``Closure and transport theory for high-collisionality electron-ion plasmas,'' \textit{Phys. Plasmas} \textbf{16}, 102108 (2009).

\bibitem{levermore1996}
C.D. Levermore, ``Moment closure hierarchies for kinetic theories,'' \textit{J. Stat. Phys.} \textbf{83}, 1021 (1996).

\end{thebibliography}
\end{document}